\documentclass[english,onecolumn,draftcls]{IEEEtran}
\usepackage[T1]{fontenc}
\usepackage[latin9]{inputenc}
\usepackage{amsthm}
\usepackage{amsmath}
\usepackage{amssymb}
\usepackage{graphicx} 
\usepackage{dsfont}
\usepackage{epstopdf} 
 
\makeatletter

\newcommand{\overbar}[1]{\mkern 1.25mu\overline{\mkern-1.25mu#1\mkern-0.25mu}\mkern 0.25mu}
\newcommand{\overbarr}[1]{\mkern 3.5mu\overline{\mkern-3.5mu#1\mkern-0.75mu}\mkern 0.75mu}
\newcommand{\mytilde}[1]{\mkern 0.5mu\widetilde{\mkern-0.5mu#1\mkern-0mu}\mkern 0mu}

\newcommand{\openone}{\mathds{1}}

\newcommand{\CLM}{C_{\mathrm{LM}}}
\newcommand{\CSC}{C_{\mathrm{SC}}}
\newcommand{\ISC}{I_{\mathrm{SC}}}
\newcommand{\CM}{C_{\mathrm{M}}}
\newcommand{\Qmin}{Q'_{\mathrm{min}}}
\newcommand{\Phat}{\widehat{P}}

\newcommand{\pebar}{\overbar{p}_{\mathrm{e}}}

\newcommand{\pe}{p_{\mathrm{e}}}

\newcommand{\LM}{I_{\mathrm{LM}}}

\newcommand{\atilde}{\tilde{a}}

\newcommand{\Ptilde}{\widetilde{P}}

\newcommand{\Wtilde}{\widetilde{W}}
\newcommand{\Wvtilde}{\widetilde{\boldsymbol{W}}}

\newcommand{\ubar}{\overbar{u}}

\newcommand{\xbar}{\overbar{x}}
\newcommand{\Xbar}{\overbarr{X}}

\newcommand{\av}{\boldsymbol{a}}

\newcommand{\uvbar}{\overbar{\boldsymbol{u}}}
\newcommand{\uv}{\boldsymbol{u}}
\newcommand{\Uvbar}{\overbar{\boldsymbol{U}}}
\newcommand{\Uv}{\boldsymbol{U}}
\newcommand{\vv}{\boldsymbol{v}}

\newcommand{\xvbar}{\overbar{\boldsymbol{x}}}
\newcommand{\xvtilde}{\widetilde{\boldsymbol{x}}}
\newcommand{\xv}{\boldsymbol{x}}

\newcommand{\Xvbar}{\overbarr{\boldsymbol{X}}}
\newcommand{\Xvtilde}{\mytilde{\boldsymbol{X}}}

\newcommand{\Xv}{\boldsymbol{X}}
\newcommand{\yv}{\boldsymbol{y}}
\newcommand{\Yv}{\boldsymbol{Y}}

\newcommand{\Ac}{\mathcal{A}}

\newcommand{\Cc}{\mathcal{C}}

\newcommand{\Pc}{\mathcal{P}}

\newcommand{\Uc}{\mathcal{U}}

\newcommand{\Xc}{\mathcal{X}}
\newcommand{\Yc}{\mathcal{Y}}

\newcommand{\EE}{\mathbb{E}}
\newcommand{\PP}{\mathbb{P}}
\newcommand{\RR}{\mathbb{R}}

\newcommand{\defeq}{\triangleq}

\newcommand{\peobar}{\overbar{p}_{\mathrm{e},0}}
\newcommand{\peibar}{\overbar{p}_{\mathrm{e},1}}

\newcommand{\Eorcc}{E_{\mathrm{r},0}^{\mathrm{cc}}}

\newcommand{\qv}{\boldsymbol{q}}
\newcommand{\Wv}{\boldsymbol{W}}

  \theoremstyle{plain} 
  
  \theoremstyle{plain}
  \newtheorem{lem}{\protect\lemmaname}
  \theoremstyle{plain} 
  \newtheorem{ce}{\protect\cename}
\usepackage{flushend}
\usepackage{cite}
\usepackage{enumerate}
\usepackage{amsmath}
\usepackage{subfig}

\makeatother

\usepackage{babel} 
  \providecommand{\lemmaname}{Lemma}
  \providecommand{\propositionname}{Proposition}
  \providecommand{\cename}{Counter-Example}

\begin{document}

\title{A Counter-Example to the Mismatched \\ Decoding Converse for Binary-Input \\ Discrete Memoryless
Channels}

\author{Jonathan Scarlett, Anelia Somekh-Baruch, Alfonso Martinez and Albert
Guill\'en i F\`abregas}

\long\def\symbolfootnote[#1]#2{\begingroup\def\thefootnote{\fnsymbol{footnote}}\footnote[#1]{#2}\endgroup}
\maketitle

\begin{abstract}
This paper studies the mismatched decoding problem for binary-input
discrete memoryless channels. An example is provided for which an
achievable rate based on superposition coding exceeds
the LM rate (Hui, 1983; Csisz\'ar-K\"orner, 1981), thus providing a counter-example
to a previously reported converse result (Balakirsky, 1995). Both
numerical evaluations and theoretical results are used in establishing
this claim. 
\end{abstract}

\symbolfootnote[0]{J.~Scarlett is with the Laboratory for Information and Inference Systems, \'Ecole Polytechnique F\'ed\'erale de Lausanne, CH-1015, Switzerland (e-mail: jmscarlett@gmail.com). 

A.~Somekh-Baruch is with the Faculty of Engineering at Bar-Ilan University, Ramat-Gan, Israel. (e-mail: somekha@biu.ac.il).

A.~Martinez is with the Department of Information and Communication Technologies,  Universitat Pompeu Fabra, 08018 Barcelona, Spain (e-mail: alfonso.martinez@ieee.org). 

A.~Guill\'en i F\`abregas is with the Instituci\'o Catalana de Recerca i Estudis  Avan\c{c}ats (ICREA), the Department of Information and Communication Technologies,  Universitat Pompeu Fabra, 08018 Barcelona, Spain, and also with the Department of  Engineering, University of Cambridge, Cambridge, CB2 1PZ, U.K. (e-mail:  guillen@ieee.org).

This work has been funded by the European Research Council under ERC 
grant agreement 259663, by the European Union's 7th Framework Programme 
under grant agreement 303633, by the Spanish Ministry of Economy and 
Competitiveness under grants RYC-2011-08150 and TEC2012-38800-C03-03,
and by the Israel Science Foundation under Grant 2013/919.

This work was presented at the Information Theory and Applications Workshop (2015),
and is an extended version of a paper accepted to the IEEE Transactions on Information Theory.}

\section{Introduction} \label{sec:MAC_INTRODUCTION}

In this paper, we consider the problem of channel coding with a given
(possibly suboptimal) decoding rule, i.e.~mismatched decoding 
\cite{Hui,Csiszar1,Csiszar2,Merhav}.  This problem is of significant
interest in settings where the optimal decoder is ruled out due to
channel uncertainty or implementation constraints, and also has
several connections to theoretical problems such as zero-error capacity.
Finding a single-letter expression for the channel capacity with mismatched
decoding is a long-standing open problem, and is believed to be very
difficult; the vast majority of the literature has focused on 
achievability results.  The only reported single-letter converse
result for general decoding metrics is that of Balakirsky \cite{Balikirsky},
who considered binary-input discrete memoryless channels (DMCs) and stated 
a matching converse to the achievable rate of Hui \cite{Hui}
and Csisz\'ar-K\"orner \cite{Csiszar1}.  However, in the present paper,
we provide a counter-example to this converse, i.e.~a binary-input
DMC for which this rate can be exceeded.

We proceed by describing the problem setup.  The encoder and decoder
share a codebook $\Cc=\{\xv^{(1)},\dotsc\xv^{(M)}\}$
containing $M$ codewords of length $n$. The encoder receives a message
$m$ equiprobable on the set $\{1,\dotsc M\}$ and transmits $\xv^{(m)}$.
The output sequence $\yv$ is generated according to $W^{n}(\yv|\xv)=\prod_{i=1}^{n}W(y_{i}|x_{i})$,
where $W$ is a single-letter transition law from $\Xc$ to
$\Yc$. The alphabets are assumed to be finite, and hence
the channel is a DMC. 
Given the output sequence $\yv$, an estimate of the message is formed as follows:
\begin{equation}
    \hat{m}=\arg\max_{j}q^{n}(\xv^{(j)},\yv),\label{eq:CNV_DecodingRule}
\end{equation}
where $q^{n}(\xv,\yv)\defeq\prod_{i=1}^{n}q(x_{i},y_{i})$
for some non-negative function $q$ called the \emph{decoding metric}.
An error is said to have occurred if $\hat{m}$ differs from $m$, and
the error probability is denoted by 
\begin{equation}
    \pe \triangleq \PP[\hat{m} \ne m].
\end{equation}
We assume that ties are broken as errors.
A rate $R$ is said to be achievable if, for all $\delta>0$,
there exists a sequence of codebooks with $M\ge e^{n(R-\delta)}$
codewords having vanishing error probability under the decoding rule
in \eqref{eq:CNV_DecodingRule}. The mismatched capacity of $(W,q)$
is defined to be the supremum of all achievable rates, and is denoted by $\CM$.

In this paper, we focus on binary-input DMCs, and we will be primarily
interested in the achievable rates based on constant-composition codes
due to Hui \cite{Hui} and Csisz\'ar and K\"orner \cite{Csiszar1}, an
achievable rate based on superposition coding by the present authors 
\cite{JournalMU,MMSomekh,Thesis}, and a reported converse by Balakirsky
\cite{Balikirsky}.  These are introduced in Sections \ref{sec:INTR_ACHIEV}
and \ref{sec:INTR_CONV}.

\subsection{Notation} \label{sub:CNV_NOTATION}

 The set of all probability mass functions
(PMFs) on a given finite alphabet, say $\Xc$, is denoted
by $\Pc(\Xc)$, and similarly for conditional distributions
(e.g. $\Pc(\Yc|\Xc)$). The marginals of a joint  distribution
$P_{XY}(x,y)$ are denoted by $P_{X}(x)$ and $P_{Y}(y)$. Similarly,
$P_{Y|X}(y|x)$ denotes the conditional distribution induced by $P_{XY}(x,y)$.
We write $P_{X}=\Ptilde_{X}$ to denote element-wise equality
between two probability distributions on the same alphabet. Expectation
with respect to a distribution $P_{X}(x)$ is denoted by $\EE_{P}[\cdot]$.
Given a distribution $Q(x)$ and a conditional distribution $W(y|x)$,
the joint distribution $Q(x)W(y|x)$ is denoted by $Q\times W$. Information-theoretic
quantities with respect to a given distribution (e.g. $P_{XY}(x,y)$)
are written using a subscript (e.g. $I_{P}(X;Y)$). All logarithms
have base $e$, and all rates are in nats/use.

\subsection{Achievability} \label{sec:INTR_ACHIEV}

The most well-known achievable rate in the literature, and the one
of the most interest in this paper, is the LM rate, which is
given as follows for an arbitrary input distribution $Q\in\Pc(\Xc)$:
\begin{equation}
    \LM(Q)\defeq\min_{\substack{\Ptilde_{XY}\in\Pc(\Xc\times\Yc)\,:\,\Ptilde_{X}=Q,\,\Ptilde_{Y}=P_{Y}\\
    \EE_{\Ptilde}[\log q(X,Y)]\ge\EE_{P}[\log q(X,Y)]}} I_{\Ptilde}(X;Y),\label{eq:CNV_PrimalLM}
\end{equation}
where $P_{XY}\defeq Q\times W$.
This rate was derived independently by Hui \cite{Hui} and Csisz\'ar-K\"orner
\cite{Csiszar1}.  The proof uses a standard random coding construction
in which each codeword is independently drawn according to the uniform distribution
on a given type class.  The following alternative expression was given
by Merhav \emph{et al.} \cite{Merhav} using Lagrange duality:
\begin{equation}
    \LM(Q)\defeq\sup_{s\ge0,a(\cdot)}\sum_{x,y}Q(x)W(y|x)\log\frac{q(x,y)^{s}e^{a(x)}}{\sum_{\xbar}Q(\xbar)q(\xbar,y)^{s}e^{a(\xbar)}}.\label{eq:CNV_DualLM}
\end{equation}
Since the input distribution $Q$ is arbitrary, we can optimize it
to obtain the achievable rate $\CLM\defeq\max_{Q}\LM(Q)$.
In general, $\CM$ may be strictly higher than $\CLM$ \cite{Csiszar1,MacMM}.

The first approach to obtaining achievable rates
exceeding $\CLM$ was given in \cite{Csiszar1}. The idea
is to code over pairs of symbols: If a rate $R$ is achievable for
the channel $W^{(2)}((y_{1},y_{2})|(x_{1},x_{2}))\defeq W(y_{1}|x_{1})W(y_{2}|x_{2})$
with the metric $q^{(2)}((x_{1},x_{2}),(y_{1},y_{2}))\defeq q(x_{1},y_{1})q(x_{2},y_{2})$,
then $\frac{R}{2}$ is achievable for the original channel $W$ with
the metric $q$. Thus, one can apply the LM rate to $(W^{(2)},q^{(2)})$,
optimize the input distribution on the product alphabet, and infer an achievable
rate for $(W,q)$; we denote this rate by $\CLM^{(2)}$.
An example was given in \cite{Csiszar1} for which $\CLM^{(2)}>\CLM$.
Moreover, as stated in \cite{Csiszar1}, the preceding arguments can
be applied to the $k$-th order product channel for $k>2$; we denote
the corresponding achievable rate by $\CLM^{(k)}$. It
was conjectured in \cite{Csiszar1} that $\lim_{k\to\infty}\CLM^{(k)}=\CM$.
It should be noted that the computation of $\CLM^{(k)}$ is generally
prohibitively complex even for relatively small values of $k$, since
$\LM(Q)$ is non-concave in general \cite{MMRevisited}.

Another approach to improving on $\CLM$ is to use multi-user
random coding ensembles exhibiting more structure than the standard
ensemble containing independent codewords. This idea was first proposed
by Lapidoth \cite{MacMM}, who used parallel coding techniques to
provide an example where $\CM=C$ (with $C$ being the
matched capacity) but $\CLM<C$. Building on these ideas,
further achievable rates were provided by the present authors \cite{MMSomekh,JournalMU,Thesis}
using superposition coding techniques. Of particular interest in this
paper is the following. For any finite auxiliary alphabet $\Uc$
and input distribution $Q_{UX}$, the rate $R=R_{0}+R_{1}$ is achievable
for any $(R_{0},R_{1})$ satisfying\footnote{The condition in \eqref{eq:CNV_RsumSC} 
has a slightly different form to that in \cite{JournalMU}, which contains the
additional constraint $I_{\Ptilde}(U;X) \le R_0$ and replaces the $[\cdot]^{+}$
function in the objective by its argument.  Both forms are given in \cite{MMSomekh}, and 
their equivalence is proved therein.  A simple way of seeing this equivalence is 
by noting that both expressions can be written as $0 \le \min_{\Ptilde_{UXY}}\max\big\{I_{\Ptilde}(U,X;Y) - (R_0+R_1), I_{\Ptilde}(U;X) - R_0 \big\}$.}
\begin{gather}
    R_{1} \le\min_{\substack{\Ptilde_{UXY}\in\Pc(\Uc\times\Xc\times\Yc)\,:\,\Ptilde_{UX}=Q_{UX},\,\Ptilde_{UY}=P_{UY}\\
    \EE_{\Ptilde}[\log q(X,Y)]\ge\EE_{P}[\log q(X,Y)]}} I_{\Ptilde}(X;Y|U)\label{eq:CNV_R1SC}\\
    R_{0} \le\min_{\substack{\Ptilde_{UXY}\in\Pc(\Uc\times\Xc\times\Yc)\,:\,\Ptilde_{UX}=Q_{UX},\,\Ptilde_{Y}=P_{Y}\\
    \EE_{\Ptilde}[\log q(X,Y)]\ge\EE_{P}[\log q(X,Y)]}} I_{\Ptilde}(U;X) + \big[I_{\Ptilde}(X;Y|U) - R_1\big]^{+}, \label{eq:CNV_RsumSC}
\end{gather}
where $P_{UXY} \triangleq Q_{UX} \times W$.
We define $I_{\mathrm{SC}}(Q_{UX})$ to be the maximum of $R_{0}+R_{1}$
subject to these constraints, and we write the optimized rate as $\CSC\defeq\sup_{\Uc,Q_{UX}}I_{\mathrm{SC}}(Q_{UX})$.
We also note the following dual expressions for \eqref{eq:CNV_R1SC}--\eqref{eq:CNV_RsumSC}
\cite{JournalMU,Thesis}: 
\begin{gather}
    R_{1}\le\sup_{s\ge0,a(\cdot,\cdot)}\sum_{u,x,y}Q_{UX}(u,x)W(y|x)\log\frac{q(x,y)^{s}e^{a(u,x)}}{\sum_{\xbar}Q_{X|U}(\xbar|u)q(\xbar,y)^{s}e^{a(u,\xbar)}}\label{eq:SC_R1_Dual}\\
    R_{0}\le\sup_{\rho_{1}\in[0,1],s\ge0,a(\cdot,\cdot)}\sum_{u,x,y}Q_{UX}(u,x)W(y|x)\log\frac{\big(q(x,y)^{s}e^{a(u,x)}\big)^{\rho_{1}}}{\sum_{\overline{u}}Q_{U}(\overline{u})\Big(\sum_{\xbar}Q_{X|U}(\xbar|\overline{u})q(\xbar,y)^{s}e^{a(\overline{u},\xbar)}\Big)^{\rho_{1}}}-\rho_{1}R_{1}.\label{eq:SC_Rsum_Dual}
\end{gather}
We outline the derivations of both the primal and dual expressions 
in Appendix \ref{sec:CNV_DERIVATIONS}. 

We note that $\CSC$ is at least as high as Lapidoth's
parallel coding rate \cite{MMSomekh,JournalMU,Thesis}, though it is not
known whether it can be strictly higher. In \cite{JournalMU}, a refined
version of superposition coding was shown to yield a rate improving
on $I_{\mathrm{SC}}(Q_{UX})$ for fixed $(\Uc,Q_{UX})$, but
the standard version will suffice for our purposes.

The above-mentioned technique of passing to the $k$-th order product
alphabet is equally valid for the superposition coding achievable
rate, and we denote the resulting achievable rate by $\CSC^{(k)}$.
The rate $\CSC^{(2)}$ will be particularly important in
this paper, and we will also use the analogous quantity $I_{\mathrm{SC}}^{(2)}(Q_{UX})$
with a fixed input distribution $Q_{UX}$. Since the input alphabet
of the product channel is $\Xc^2$, one might
more precisely write the input distribution as $Q_{UX^{(2)}}$, but
we omit this additional superscript.  The choice $\Uc=\{0,1\}$ for
the auxiliary alphabet will prove to be sufficient for our purposes.

\subsection{Converse} \label{sec:INTR_CONV}

Very few converse results have been provided for the mismatched decoding
problem. Csisz\'ar and Narayan \cite{Csiszar2} showed that $\lim_{k\to\infty}\CLM^{(k)}=\CM$
for erasures-only metrics, i.e.~metrics such that $q(x,y)=\max_{x,y}q(x,y)$ for all
$(x,y)$ such that $W(y|x)>0$. More recently, multi-letter converse
results were given by Somekh-Baruch \cite{MMGeneralFormula}, yielding a general
formula for the mismatched capacity in the sense of Verd\'u-Han \cite{VerduHan}.
However, these expressions are not computable.

The only general single-letter converse result presented in the 
literature is that of Balakirsky \cite{ConverseMM}, who reported 
that $\CLM=\CM$ for binary-input DMCs. In the following section, we 
provide a counter-example showing that in
fact the strict inequality $\CM>\CLM$ can hold
even in this case. 

\section{The Counter-Example}

The main claim of this paper is the following; the details are given
in Section \ref{sec:CNV_PROOF}.

\begin{ce} \label{prop:CNV_MainResult}
    Let $\Xc=\{0,1\}$ and $\Yc=\{0,1,2\}$,
    and consider the channel and metric described by the entries of the
    $|\Xc|\times|\Yc|$ matrices
    \begin{align}
        \Wv & =\left[\begin{array}{ccc}
        0.97 & 0.03 & 0\\
        0.1 & 0.1 & 0.8
        \end{array}\right],\qquad\qv=\left[\begin{array}{ccc}
        1 & 1 & 1\\
        1 & 0.5 & 1.36
        \end{array}\right].\label{eq:CNV_Channel}
    \end{align}
    Then the LM rate satisfies
    \begin{equation}
        0.136874 \le \CLM \le 0.136900 \quad\mathrm{nats/use},\label{eq:CNV_BoundLM}
    \end{equation}
    whereas the superposition coding rate obtained by considering the
    second-order product of the channel is lower bounded by
    \begin{equation}
        \CSC^{(2)}\ge0.137998\quad\mathrm{nats/use}.\label{eq:CNV_BoundSC2}
    \end{equation}
    Consequently, we have $\CM>\CLM$.
\end{ce}

\begin{figure}
    \begin{centering}
    % \subfloat{\includegraphics[width=0.75\paperwidth]{../../../MATLAB/Plots/17-BinaryInputFull.eps}}
    % \subfloat{\includegraphics[width=0.35\paperwidth]{../../../MATLAB/Plots/17-BinaryInputZoomed.eps}}
    \includegraphics[width=0.7\paperwidth]{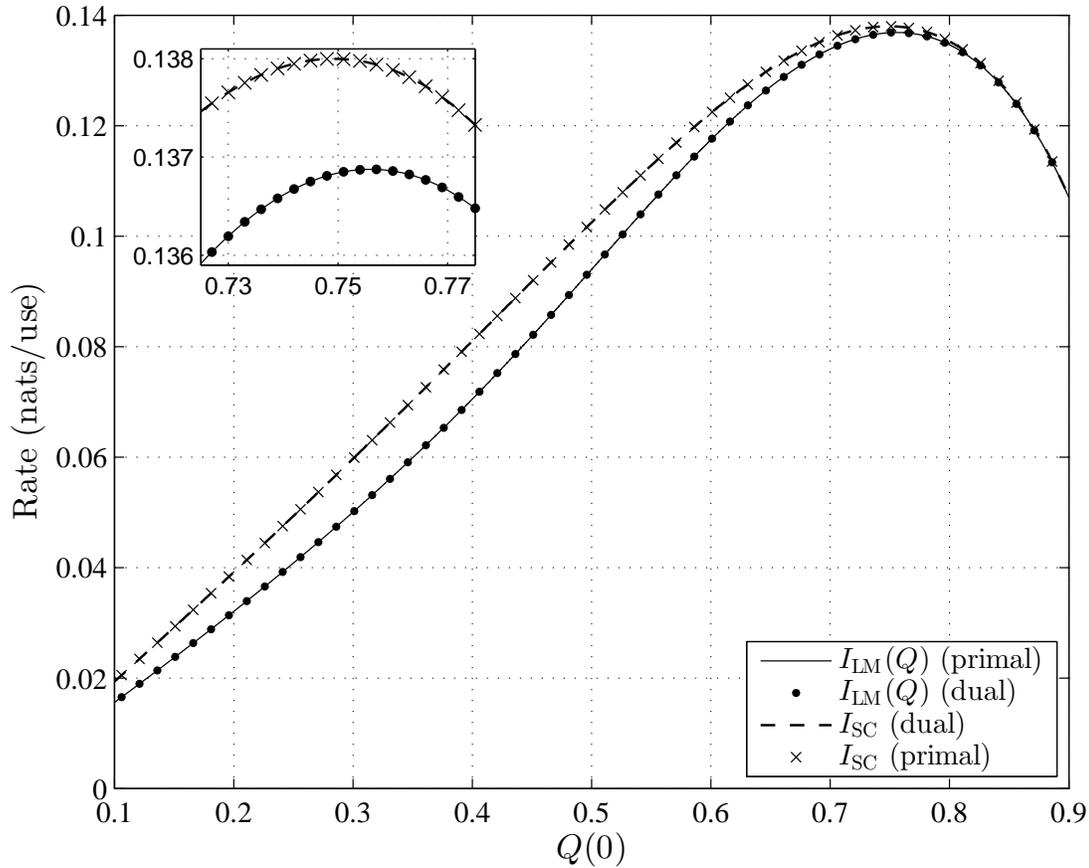}
    \par
    \end{centering}
    
    \caption{Numerical evaluations of the LM rate $\LM(Q)$ as a function
    of the (first entry of the) input distribution, and the corresponding superposition
    coding rate $I_{\mathrm{SC}}^{(2)}(Q_{UX})$ using the construction described
    in Section \ref{sub:CNV_SC_EVAL}.  The matched capacity is
    $C \approx 0.4944$ nats/use, and is achieved by $Q(0) \approx 0.5398$. \label{fig:CNV_MainPlot} }
\end{figure}
 
We proceed by presenting various points of discussion.

\paragraph*{Numerical Evaluations}

While \eqref{eq:CNV_BoundLM} and \eqref{eq:CNV_BoundSC2} are obtained
using numerical computations, and the difference between the two is
small, we will take care in ensuring that the gap is genuine,
rather than being a matter of numerical accuracy.  All of the code
used in our computations is available online \cite{PaperBI_Code}.

Figure \ref{fig:CNV_MainPlot} plots our numerical evaluations
of $\LM(Q)$ and $\ISC^{(2)}(Q_{UX})$ for a range of input 
distributions; for the latter, $Q_{UX}$ is determined from $Q$ in a manner to 
be described in Section \ref{sub:CNV_SC_EVAL}.
Note that this plot is only meant to help the reader visualize the
results; it is not sufficient to establish Counter-Example \ref{prop:CNV_MainResult}
in itself. Nevertheless, it is reassuring to see that the curves corresponding
to the primal and dual expressions are indistinguishable. 

Our computations suggest that
\begin{equation}
\CLM\approx0.136875\quad\mathrm{nats/use},\label{eq:CNV_ApproxLM}
\end{equation}
and that the optimal input distribution is approximately
\begin{equation}
    Q=\left[\begin{array}{cc}
    0.75597 & 0.24403
    \end{array}\right].\label{eq:CNV_OptQ}
\end{equation}
The matched capacity is significantly higher than $\CLM$, namely 
$C \approx 0.4944$ nats/use, with a corresponding input distribution
approximately equal to $[0.5398 ~ 0.4602]$.
As seen in the proof, the fact that the right-hand side of \eqref{eq:CNV_BoundLM}
exceeds that of \eqref{eq:CNV_ApproxLM} by $2.5 \times 10^{-5}$ is due to the use of (possibly
crude) bounds on the loss in the rate when $Q$ is slightly suboptimal.

\paragraph*{Other Achievable Rates}

One may question whether \eqref{eq:CNV_BoundSC2} can be improved
by considering $\CSC^{(k)}$ for $k>2$.  However, we were unable to find 
any such improvement when we tried $k=3$; see Section \ref{sub:CNV_SC_EVAL}
for further discussion on this attempt.  Similarly, we observed no improvement
on \eqref{eq:CNV_ApproxLM} when we computed $\LM^{(2)}(Q^{(2)})$ with a 
brute force search over $Q^{(2)}\in\Pc(\Xc^2)$ to two decimal places.
Of course, it may still be that $\CLM^{(k)} > \CLM$ for some $k>2$, but
optimizing $Q^{(k)}$ quickly becomes computationally difficult; even for 
$k=3$, the search space is $7$-dimensional with no apparent convexity structure.

Our numerical findings also showed no improvement of the superposition
coding rate $\CSC$ for the original channel (as opposed to the product
channel) over the LM rate $\CLM$.

We were also able to obtain the achievable rate in \eqref{eq:CNV_BoundLM}
using Lapidoth's expurgated parallel coding rate \cite{MacMM} (or more
precisely, its dual formulation from \cite{JournalMU}) to the second-order
product channel.  In fact, this was
done by taking the input distribution $Q_{UX}$ and the dual parameters 
$(s,a,\rho_1)$ used in \eqref{eq:SC_R1_Dual}--\eqref{eq:SC_Rsum_Dual} (see
Section \ref{sub:CNV_SC_EVAL}),
and ``transforming'' them into parameters for the expurgated parallel 
coding ensemble that achieve an identical rate.  
Details are given in Appendix \ref{sec:CNV_MAC_RATE}.

\paragraph*{Choices of Channel and Metric}

While the decoding metric in \eqref{eq:CNV_Channel} may appear to
be unusual, it should be noted that any decoding metric with $\max_{x,y}q(x,y)>0$
is equivalent to another metric yielding a matrix of this form with
the first row and first column equal to one \cite{Balikirsky,ConverseMM}.

One may question whether the LM rate can be improved for binary-input
binary-output channels, as opposed to our ternary-output example.  
However, this is not possible, since for any such channel the LM rate
is either equal to zero or the matched capacity, and in either case
it coincides with the mismatched capacity \cite{Csiszar2}.

Unfortunately, despite considerable effort, we have been unable to
understand the analysis given in \cite{ConverseMM} in sufficient
detail to identify any major errors therein. We also remark that for
the vast majority of the examples we considered, $\CLM$
was indeed greater than or equal to all other achievable rates that
we computed. However, \eqref{eq:CNV_Channel} was not the only counter-example,
and others were found with $\min_{x,y}W(y|x)>0$ (in contrast with 
\eqref{eq:CNV_Channel}).  For example, a similar gap between the rates 
was observed when the first row of $\Wv$ in \eqref{eq:CNV_Channel} was replaced
by $[0.97 ~ 0.02 ~ 0.01]$.

\section{Establishing Counter-Example \ref{prop:CNV_MainResult}} \label{sec:CNV_PROOF}

While Counter-Example \ref{prop:CNV_MainResult} is concerned with the specific channel and metric
given in \eqref{eq:CNV_Channel}, we will present several results
for more general channels with $\Xc=\{0,1\}$ and $\Yc=\{0,1,2\}$
(and in some cases, arbitrary finite alphabets). To make some of the
expressions more compact, we define $Q_{x}\defeq Q(x)$, $W_{xy}\defeq W(y|x)$
and $q_{xy}\defeq q(x,y)$ throughout this section.

\subsection{Auxiliary Lemmas}

The optimization of $\LM(Q)$ over $Q$ can be difficult,
since $\LM(Q)$ is non-concave in $Q$ in general \cite{MMRevisited}.
Since we are considering the case $|\Xc|=2$, this optimization
is one-dimensional, and we thus resort to a straightforward brute-force search
of $Q_{0}$ over a set of regularly-spaced points in $[0,1]$. To
establish the upper bound in \eqref{eq:CNV_BoundLM}, we must bound
the difference $\CLM-\LM(Q_0)$ for the choice
of $Q_{0}$ maximizing the LM rate among all such points. Lemma \ref{lem:CNV_DiffLM}
below is used for precisely this purpose; before stating it, we present
a preliminary result on the continuity of the binary entropy function
$H_{2}(\alpha)\defeq-\alpha\log\alpha-(1-\alpha)\log(1-\alpha)$.

It is well-known that for two distributions $Q$ and $Q^{\prime}$
on a common finite alphabet, we have $|H(Q^{\prime})-H(Q)|\le\delta\log\frac{|\Xc|}{\delta}$
whenever $\|Q^{\prime}-Q\|_{1}\le\delta$ \cite[Lemma 2.7]{CsiszarBook}.
The following lemma gives a refinement of this statement for the case
that $|\Xc|=2$ and $\min\{Q_{0}^{\prime},Q_{1}^{\prime}\}$
is no smaller than a predetermined constant. 
\begin{lem} \label{lem:CNV_DiffH}
    Let $Q^{\prime}\in\Pc(\Xc)$
    be a PMF on $\Xc=\{0,1\}$ such that $\min\{Q_{0}^{\prime},Q_{1}^{\prime}\}\ge \Qmin$
    for some $\Qmin>0$. For any PMF $Q\in\Pc(\Xc)$
    such that $|Q_{0}-Q_{0}^{\prime}|\le\delta$ (or equivalently, $|Q_{1}-Q_{1}^{\prime}|\le\delta$),
    we have
    \begin{equation}
    \big|H(Q^{\prime})-H(Q)\big|\le\delta\log\frac{1-\Qmin}{\Qmin}.
    \end{equation}
\end{lem}
\begin{proof}
    Set $\Delta\defeq Q_{0}-Q_{0}^{\prime}$. Since $H_{2}(\cdot)$
    is concave, the straight line tangent to a given point always lies
    above the function itself. Assuming without loss of generality that
    $Q_{0}^{\prime}\le0.5$, we have
    \begin{align}
        \big|H_{2}(Q_{0}^{\prime}+\Delta)-H_{2}(Q_{0}^{\prime})\big| & \le|\Delta|\cdot\frac{dH_{2}}{d\alpha}\bigg|_{\alpha=Q_{0}^{\prime}}\\
         & =|\Delta|\log\frac{1-Q_{0}^{\prime}}{Q_{0}^{\prime}}.
    \end{align}
    The desired result follows since $\frac{1-Q_{0}^{\prime}}{Q_{0}^{\prime}}$
    is decreasing in $Q_{0}^{\prime}$, and since $Q_{0}^{\prime}\ge \Qmin$
    and $|\Delta|\le\delta$ by assumption.
\end{proof}
The following lemma builds on the preceding lemma, and is key to  
establishing Counter-Example \ref{prop:CNV_MainResult}.
\begin{lem}
    \label{lem:CNV_DiffLM}For any binary-input mismatched DMC, we have
    the following under the setup of Lemma \ref{lem:CNV_DiffH}: 
    \begin{equation}
        \LM(Q)\ge \LM(Q^{\prime})-\delta\log\frac{1-\Qmin}{\Qmin}-\frac{\delta\log2}{\Qmin}.\label{eq:CNV_DiffLM}
    \end{equation}
\end{lem}
\begin{proof} 
    The bound in \eqref{eq:CNV_DiffLM} is trivial when $\LM(Q^{\prime})=0$,
    so we consider the case $\LM(Q^{\prime})>0$.  
    Observing that $Q(x)>0$ for $x \in \{0,1\}$, we 
    can make the change of variable $a(x)=\log\frac{e^{\atilde(x)}}{Q(x)}$ 
    (i.e. $e^{\atilde(x)}=Q(x)e^{a(x)}$) in \eqref{eq:CNV_DualLM} to obtain
    \begin{equation}
    \LM(Q)=\sup_{s\ge0,\atilde(\cdot)}\sum_{x,y}Q(x)W(y|x)\log\frac{q(x,y)^{s}e^{\atilde(x)}}{Q(x)\sum_{\xbar}q(\xbar,y)^{s}e^{\atilde(\xbar)}}, \label{eq:CNV_DualLM_Alt}
    \end{equation}
    which can equivalently be written as
    \begin{equation}
        \LM(Q^{\prime})=H(Q^{\prime})-\inf_{s\ge0,\atilde(\cdot)}\sum_{x,y}Q^{\prime}(x)W(y|x)\log\bigg(1+\frac{q(\xbar,y)^{s}e^{\atilde(\xbar)}}{q(x,y)^{s}e^{\atilde(x)}}\bigg),\label{eq:CNV_DualLM_Alt2}
    \end{equation}
    where $\xbar\in\{0,1\}$ denotes the unique symbol differing
    from $x\in\{0,1\}$. 
    
    The following arguments can be simplified when
    the infimum is achieved, but for completeness we consider the general
    case. Let $(s_{k},\atilde_{k})$ be a sequence of parameters such
    that
    \begin{equation}
        H(Q^{\prime})-\lim_{k\to\infty}\sum_{x,y}Q^{\prime}(x)W(y|x)\log\bigg(1+\frac{q(\xbar,y)^{s_{k}}e^{\atilde_{k}(\xbar)}}{q(x,y)^{s_{k}}e^{\atilde_{k}(x)}}\bigg)=\LM(Q^{\prime}).\label{eq:CNV_Lim1}
    \end{equation}
    Since the argument to the logarithm in \eqref{eq:CNV_Lim1} is no
    smaller than one, and since $H(Q^{\prime})\le\log2$ by the assumption
    that the input alphabet is binary, we have for $x=0,1$ and sufficiently
    large $k$ that
    \begin{equation}
        \sum_y Q^{\prime}(x)W(y|x)\log\bigg(1+\frac{q(\xbar,y)^{s_{k}}e^{\atilde_{k}(\xbar)}}{q(x,y)^{s_{k}}e^{\atilde_{k}(x)}}\bigg)\le\log2,\label{eq:CNV_Lim2}
    \end{equation}
    since otherwise the left-hand side of \eqref{eq:CNV_Lim1} would be
    non-positive, in contradiction with the fact that we are considering
    the case $\LM(Q^{\prime})>0$. Using the assumption $\min\{Q_{0}^{\prime},Q_{1}^{\prime}\}\ge \Qmin$,
    we can weaken \eqref{eq:CNV_Lim2} to
    \begin{equation}
        \sum_y W(y|x)\log\bigg(1+\frac{q(\xbar,y)^{s_{k}}e^{\atilde_{k}(\xbar)}}{q(x,y)^{s_{k}}e^{\atilde_{k}(x)}}\bigg)\le\frac{\log2}{\Qmin}.\label{eq:CNV_Lim3}
    \end{equation}
    We now have the following:
    \begin{align}
        \LM(Q) & \ge H(Q)-\limsup_{k\to\infty}\sum_{x,y}Q(x)W(y|x)\log\bigg(1+\frac{q(\xbar,y)^{s_{k}}e^{\atilde_{k}(\xbar)}}{q(x,y)^{s_{k}}e^{\atilde_{k}(x)}}\bigg)\label{eq:CNV_Lim4}\\
         & \ge H(Q^{\prime})-\limsup_{k\to\infty}\sum_{x}Q(x)\sum_{y}W(y|x)\log\bigg(1+\frac{q(\xbar,y)^{s_{k}}e^{\atilde_{k}(\xbar)}}{q(x,y)^{s_{k}}e^{\atilde_{k}(x)}}\bigg)-\delta\log\frac{1-\Qmin}{\Qmin}\label{eq:CNV_Lim5}\\
         & = H(Q^{\prime})-\limsup_{k\to\infty}\sum_{x}(Q(x) + Q'(x) - Q'(x))\sum_{y}W(y|x)\log\bigg(1+\frac{q(\xbar,y)^{s_{k}}e^{\atilde_{k}(\xbar)}}{q(x,y)^{s_{k}}e^{\atilde_{k}(x)}}\bigg)-\delta\log\frac{1-\Qmin}{\Qmin}\label{eq:CNV_Lim5a}\\
         & \ge H(Q^{\prime})-\limsup_{k\to\infty}\sum_{x}Q^{\prime}(x)\sum_{y}W(y|x)\log\bigg(1+\frac{q(\xbar,y)^{s_{k}}e^{\atilde_{k}(\xbar)}}{q(x,y)^{s_{k}}e^{\atilde_{k}(x)}}\bigg)-\delta\log\frac{1-\Qmin}{\Qmin}-\frac{\delta\log2}{\Qmin}\label{eq:CNV_Lim6}\\
         & =\LM(Q^{\prime})-\delta\log\frac{1-\Qmin}{\Qmin}-\frac{\delta\log2}{\Qmin},\label{eq:CNV_Lim7}
    \end{align}
    where \eqref{eq:CNV_Lim4} follows by replacing the infimum in \eqref{eq:CNV_DualLM_Alt2}
    by the particular sequence of parameters $(s_{k},\atilde_{k})$ and
    taking the $\limsup$, \eqref{eq:CNV_Lim5} follows
    from Lemma \ref{lem:CNV_DiffH}, \eqref{eq:CNV_Lim6} follows by applying
    \eqref{eq:CNV_Lim3} for the $x$ value where $Q(x) \ge Q'(x)$ and lower
    bounding the logarithm by zero for the other $x$ value, and \eqref{eq:CNV_Lim7} 
    follows from \eqref{eq:CNV_Lim1}. 
\end{proof}

\subsection{Establishing the Upper Bound in \eqref{eq:CNV_BoundLM}}

As mentioned in the previous subsection, we optimize $Q$ by
performing a brute force search over a set of regularly spaced points,
and then using Lemma \ref{lem:CNV_DiffLM} to bound the difference
$\CLM-\LM(Q)$. We let the input distribution
therein be $Q^{\prime}=\arg\max_{Q}\LM(Q)$.  Note that this maximum
is always achieved, since $\LM$ is continuous and bounded \cite{Csiszar2}. 
If there are multiple maximizers, we choose one arbitrarily among them. 

To apply Lemma \ref{lem:CNV_DiffLM}, we need a constant $\Qmin$ such that
$\min\{Q_{0}^{\prime},Q_{1}^{\prime}\}\ge \Qmin$. 
We present a straightforward choice based on the lower bound on the left-hand
side of \eqref{eq:CNV_BoundLM} (proved in Section \ref{eq:CNV_PfLowerLM}).
By choosing $\Qmin$ such that even the mutual information $I(X;Y)$ is upper
bounded by the left-hand side of \eqref{eq:CNV_BoundLM} when 
$\min\{Q_{0}^{\prime},Q_{1}^{\prime}\}<\Qmin$,
we see from the simple identity $\LM(Q)\le I(X;Y)$ \cite{Csiszar2}
that $Q$ cannot maximize $\LM$. For the example under consideration
(see \eqref{eq:CNV_Channel}), the choice $\Qmin=0.042$
turns out to be sufficient, and in fact yields $I(X;Y)\le0.135$.
This can be verified by computing $I(X;Y)$ to be (approximately)
$0.0917$, $0.4919$ and $0.1348$ for $Q_0 = 0.042$, $Q_0 = 0.5$
and $Q_0 = 1 - 0.042$ respectively, and then using the concavity
of $I(X;Y)$ in $Q$ to handle $Q_0 \in [0,0.042) \cup (1-0.042,1]$.

Let $h\defeq10^{-5}$, and suppose that we evaluate $\LM(Q)$
for each $Q_{0}$ in the set
\begin{equation}
    \Ac\defeq\big\{ \Qmin,\Qmin+h,\dotsc,1-\Qmin-h,1-\Qmin\big\}.
\end{equation}
Since the optimal input distribution $Q^{\prime}$ corresponds to
some $Q_{0}^{\prime}\in[\Qmin,1-\Qmin]$, we conclude
that there exists some $Q_{0}\in\Ac$ such that $|Q_{0}^{\prime}-Q_{0}|\le\frac{h}{2}$.
Substituting $\delta=\frac{h}{2}=0.5\times10^{-5}$ and $\Qmin=0.042$
into \eqref{eq:CNV_DiffLM}, we conclude that 
\begin{equation}
    \max_{Q_{0}\in\Ac}\LM(Q)\ge \CLM - 0.982 \times 10^{-4}.\label{eq:CNV_DiffBound}
\end{equation}

We now describe our techniques for evaluating $\LM(Q)$
for a fixed choice of $Q$. This is straightforward in principle,
since the corresponding optimization problem is convex whether we
use the primal expression in \eqref{eq:CNV_PrimalLM} or the dual
expression in \eqref{eq:CNV_DualLM}. Nevertheless, since we need
to test a large number of $Q_{0}$ values, we make an effort to find
a reasonably efficient method.

We avoid using the dual expression in \eqref{eq:CNV_DualLM}, since
it is a \emph{maximization }problem; thus, if the final optimization
parameters obtained differ slightly from the true optimal parameters,
they will only provide a lower bound on $\LM(Q)$. 
In contrast, the result that we seek is an upper
bound. We also avoid evaluating \eqref{eq:CNV_PrimalLM} directly,
since the equality constraints in the optimization problem could, in principle,
be sensitive to numerical precision errors.

Of course, there are many ways to circumvent these problems and provide
rigorous bounds on the suboptimality of optimization procedures, including
a number of generic solvers. We instead take a different approach,
and reduce the primal optimization in \eqref{eq:CNV_BoundLM} to a
\emph{scalar minimization} problem by eliminating the constraints
one-by-one.  This minimization will contain no equality constraints,
and thus minor variations in the optimal parameter will still produce
a valid upper bound. 

We first note that the inequality constraint can be replaced by an
equality whenever $\LM(Q)>0$ \cite[Lemma 1]{Csiszar2},
which is certainly the case for the present example. Moreover, since
the $X$-marginal is constrained to equal $Q$, we can let the minimization
be over $\Pc(\Yc|\Xc)$ instead of $\Pc(\Xc \times \Yc)$, yielding
\begin{equation}
    \LM(Q)=\min_{\substack{\Wtilde\in\Pc(\Yc|\Xc)\,:\,\Ptilde_{Y}=P_{Y}\\
    \EE_{Q\times\Wtilde}[\log q(X,Y)]=\EE_{P}[\log q(X,Y)]}} I_{Q\times\Wtilde}(X;Y),\label{eq:CNV_PrimalLM_eq}
\end{equation}
where $\Ptilde_{Y}(y)\defeq\sum_{x}Q(x)\Wtilde(y|x)$ (recall
also that $P_{XY} = Q \times W$).
Let us fix a conditional distribution $\Wtilde$ satisfying
the specified constraints, and write $\Wtilde_{xy}\defeq\Wtilde(y|x)$.
The analogous matrix to $\Wv$ in \eqref{eq:CNV_Channel}
can be written as follows:
\begin{equation}
    \Wvtilde=\left[\begin{array}{ccc}
    \Wtilde_{00} & \Wtilde_{01} & 1-\Wtilde_{00}-\Wtilde_{01}\\
    \Wtilde_{10} & \Wtilde_{11} & 1-\Wtilde_{10}-\Wtilde_{11}
    \end{array}\right].
\end{equation}
Since $\Ptilde_{Y}=P_{Y}$ implies $H(\Ptilde_{Y})=H(P_{Y})$,
we can write the objective in \eqref{eq:CNV_PrimalLM_eq} as
\begin{align}
    I_{Q\times\Wtilde}(X;Y) & =H(P_{Y})-H_{Q\times\Wtilde}(Y|X)\\
        & =H(P_{Y})+Q_{0}\big(\Wtilde_{00}\log\Wtilde_{00}+\Wtilde_{01}\log\Wtilde_{01}+(1-\Wtilde_{00}-\Wtilde_{01})\log(1-\Wtilde_{00}-\Wtilde_{01})\big)\nonumber \\
        & \qquad\qquad~+Q_{1}\big(\Wtilde_{10}\log\Wtilde_{10}+\Wtilde_{11}\log\Wtilde_{11}+(1-\Wtilde_{10}-\Wtilde_{11})\log(1-\Wtilde_{10}-\Wtilde_{11})\big). \label{eq:CNV_Obj2}
\end{align}
We now show that the equality constraints can be used to express each
$\Wtilde_{xy}$ in terms of $\Wtilde_{10}$. Using $\Ptilde_{Y}(y)=P_{Y}(y)$
for $y=0,1$, along with the constraint containing the decoding metric,
we have
\begin{gather}
    Q_{0}\Wtilde_{00}+Q_{1}\Wtilde_{10}=P_{Y}(0)\label{eq:CNV_Constr1}\\
    Q_{0}\Wtilde_{01}+Q_{1}\Wtilde_{11}=P_{Y}(1)\\
    Q_{1}\big(\Wtilde_{11}\log q_{11}+(1-\Wtilde_{10}-\Wtilde_{11})\log q_{12}\big)=\EE_{P}[\log q(X,Y)],\label{eq:CNV_Constr3}
\end{gather}
where in \eqref{eq:CNV_Constr3} we used the fact that $\log q(x,y)=0$
for four of the six $(x,y)$ pairs (see \eqref{eq:CNV_Channel}).
Re-arranging \eqref{eq:CNV_Constr1}--\eqref{eq:CNV_Constr3}, we
obtain
\begin{align}
    \Wtilde_{00} & =\frac{P_{Y}(0)-Q_{1}\Wtilde_{10}}{Q_{0}}\label{eq:CNV_Constr1a}\\
    \Wtilde_{01} & =\frac{P_{Y}(1)-Q_{1}\Wtilde_{11}}{Q_{0}}\label{eq:CNV_Constr2a}\\
    \Wtilde_{11} & =\frac{1}{\log q_{11}-\log q_{12}}\bigg(\frac{\EE_{P}[\log q(X,Y)]}{Q_{1}}-(1-\Wtilde_{10})\log q_{12}\bigg),\label{eq:CNV_Constr3a}
\end{align}
and substituting \eqref{eq:CNV_Constr3a} into \eqref{eq:CNV_Constr2a}
yields
\begin{equation}
    \Wtilde_{01}=\frac{1}{Q_{0}}\Bigg(P_{Y}(1)-\frac{1}{\log q_{11}-\log q_{12}}\bigg(\EE_{P}[\log q(X,Y)]-Q_{1}(1-\Wtilde_{10})\log q_{12}\bigg)\Bigg).\label{eq:CNV_Constr4a}
\end{equation}
We have thus written each entry of \eqref{eq:CNV_Obj2} in terms of
$\Wtilde_{10}$, and we are left with a one-dimensional optimization
problem. However, we must still ensure that the constraints $\Wtilde_{xy}\in[0,1]$
are satisfied for all $(x,y)$. Since each $\Wtilde_{xy}$ is
an affine function of $\Wtilde_{10}$, these constraints are
each of the form $\underline{W}^{(x,y)}\le\Wtilde_{10}\le\overline{W}^{(x,y)}$,
and the overall optimization is given by
\begin{equation}
    \min_{\underline{W}\le\Wtilde_{10}\le\overline{W}}f(\Wtilde_{10}),\label{eq:CNV_SimplerOpt}
\end{equation}
where $f(\cdot)$ denotes the right-hand side of \eqref{eq:CNV_Obj2}
upon substituting \eqref{eq:CNV_Constr1a}, \eqref{eq:CNV_Constr3a}
and \eqref{eq:CNV_Constr4a}, and the lower and upper limits are given
by $\underline{W}\defeq\max_{x,y}\underline{W}^{(x,y)}$ and $\overline{W}\defeq\min_{x,y}\overline{W}^{(x,y)}$.
Note that the minimization region is non-empty, since $\Wvtilde=\Wv$
is always feasible. In principle one could observe $\underline{W}=\overline{W}=W_{10}$,
but in the present example we found that $\underline{W}<\overline{W}$
for every choice of $Q_{0}$ that we used.

The optimization problem in \eqref{eq:CNV_SimplerOpt} does not appear
to permit an explicit solution. However, we can efficiently compute
the solution to high accuracy using standard one-dimensional optimization
methods.  Since the convexity of any optimization problem is preserved by the 
elimination of equality constraints \cite[Sec.~4.2.4]{Convex}, and
since the optimization problem in \eqref{eq:CNV_PrimalLM_eq} is convex for
any given $Q$, we conclude that $f(\cdot)$ is a convex function.
Its derivative is easily computed by noting that
\begin{equation}
    \frac{d}{dz}(\alpha z+\beta)\log(\alpha z+\beta)=\alpha+\alpha\log(\alpha z+\beta)
\end{equation}
for all $\alpha$, $\beta$ and $z$ yielding a positive argument
to the logarithm. We can thus perform a bisection search as follows,
where $f^{\prime}(\cdot)$ denotes the derivative of $f$, and $\epsilon$
is a termination parameter:
\begin{enumerate}
    \item Set $i=0$, $\underline{W}^{(0)}=\underline{W}$ and $\overline{W}^{(0)}=\overline{W}$;
    \item Set $W_{\mathrm{mid}}=\frac{1}{2}(\underline{W}^{(i)}+\overline{W}^{(i)})$;
    if $f^{\prime}(W_{\mathrm{mid}})\ge0$ then set $\underline{W}^{(i+1)}=\underline{W}^{(i)}$
    and $\overline{W}^{(i+1)}=W_{\mathrm{mid}}$; otherwise set $\underline{W}^{(i+1)}=W_{\mathrm{mid}}$
    and $\overline{W}^{(i+1)}=\overline{W}^{(i)}$;
    \item If $|f^{\prime}(W_{\mathrm{mid}})|\le\epsilon$ then terminate; otherwise
    increment $i$ and return to Step 2.
\end{enumerate}
As mentioned previously, we do not need to find the exact solution to \eqref{eq:CNV_SimplerOpt},
since any value of $\Wtilde_{10}\in[\underline{W},\overline{W}]$
yields a valid upper bound on $\LM(Q)$. However, we must
choose $\epsilon$ sufficiently small so that the bound in \eqref{eq:CNV_BoundLM}
is established. We found $\epsilon=10^{-6}$ to suffice.

We implemented the preceding techniques in C (see \cite{PaperBI_Code} 
for the code) to upper bound $\LM(Q)$
for each $Q_{0}\in\Ac$; see Figure \ref{fig:CNV_MainPlot}.  As stated
following Counter-Example \ref{prop:CNV_MainResult}, we found the highest
value of $\LM(Q)$ to be the right-hand side of \eqref{eq:CNV_ApproxLM},
corresponding to the input distribution in \eqref{eq:CNV_OptQ}.  We found
the corresponding minimizing parameter in \eqref{eq:CNV_SimplerOpt} to 
be roughly $\Wtilde_{10} = 0.4252347$.

Instead of directly adding $10^{-4}$ to \eqref{eq:CNV_ApproxLM} in
accordance with \eqref{eq:CNV_DiffBound}, we obtain a refined estimate
by ``updating'' our estimate of $\Qmin$.  Specifically, using 
\eqref{eq:CNV_DiffBound} and observing the values in Figure 
\ref{fig:CNV_MainPlot}, we can conclude that the optimal value of
$Q_0$ lies in the range $[0.7,0.8]$ (we are being highly conservative
here).  Thus, setting $\Qmin=0.2$ and using the previously chosen 
value $\delta = 0.5 \times 10^{-5}$, we obtain the following refinement
of \eqref{eq:CNV_DiffBound}:
\begin{equation}
    \max_{Q_{0}\in\Ac}\LM(Q)\ge \CLM - 2.43 \times 10^{-5}. \label{eq:CNV_DiffBound2}
\end{equation}

Since our implementation in C is based on floating-point calculations, the
final values may have precision errors.  We therefore checked our numbers using Mathematica's 
arbitrary-precision arithmetic framework \cite{MathematicaPrecision},
which allows one to work with \emph{exact} expressions that can then be
displayed to arbitrarily many decimal places.  More precisely, we loaded
the values of $\Wtilde_{10}$ into Mathematica and rounded them to 12 decimal
places (this is allowed, since any value of $\Wtilde_{10}$ yields a valid 
upper bound).  Using the exact values of all other quantities (e.g.~$Q$ and $W$),
we performed an evaluation of $f(\Wtilde_{10})$ in \eqref{eq:CNV_SimplerOpt},
and compared it to the corresponding value of $\LM(Q)$ produced by the C 
program.  The maximum discrepancy across all of the values of $Q_0$ was 
less than $2.1 \times 10^{-12}$.  
Our final bound in \eqref{eq:CNV_BoundLM} was obtained by adding
$2.5 \times 10^{-5}$ (which is, of course, higher than 
$2.43 \times 10^{-5} + 2.1 \times 10^{-12}$) to 
the right-hand side of \eqref{eq:CNV_ApproxLM}. 

\subsection{Establishing the Lower Bound in \eqref{eq:CNV_BoundLM}} \label{eq:CNV_PfLowerLM}

For the lower bound, we can afford to be less careful than we were in establishing
the upper bound; all we need is a suitable choice of $Q$
and the parameters $(s,a)$ in \eqref{eq:CNV_DualLM}.  We choose $Q$ as in \eqref{eq:CNV_OptQ}, 
along with the following:
\begin{align}
    s &= 9.031844\\
    \av & =\left[\begin{array}{cccc}
    0.355033 & -0.355033
    \end{array}\right],
\end{align}
In Appendix \ref{sec:CNV_FURTHER_TECHNIQUES}, we provide details on how these
parameters were obtained, though the desired lower bound can readily be verified
without knowing such details.

Using these values, we evaluated the objective in \eqref{eq:CNV_DualLM} using
Mathematica's arbitrary-precision arithmetic framework \cite{MathematicaPrecision},
thus eliminating the possibility of arithmetic precision errors.
See \cite{PaperBI_Code} for the relevant C and Mathematica code.

\subsection{Establishing the Lower Bound in \eqref{eq:CNV_BoundSC2}} \label{sub:CNV_SC_EVAL}

We establish the lower bound in \eqref{eq:CNV_BoundSC2} by setting
$\Uc=\{0,1\}$ and forming a suitable choice of $Q_{UX}$,
and then using the dual expressions in \eqref{eq:SC_R1_Dual}--\eqref{eq:SC_Rsum_Dual}
to lower bound $\ISC^{(2)}(Q_{UX})$. 

\subsubsection{Choice of Input Distribution}

Let $Q=[Q_{0}\,\, Q_{1}]$ be some input distribution on $\Xc$,
and define the corresponding product distribution on $\Xc^2$ as
\begin{equation}
    Q^{(2)}=\left[\begin{array}{cccc}
    Q_{0}^{2} & Q_{0}Q_{1} & Q_{0}Q_{1} & Q_{1}^{2}
    \end{array}\right],\label{eq:CNV_OptQ2}
\end{equation}
where the order of the inputs is $(0,0)$, $(0,1)$, $(1,0)$, $(1,1)$.
Consider now the following choice of superposition coding parameters
for the second-order product channel $(W^{(2)},q^{(2)})$:
\begin{align}
    Q_{U} & =\left[\begin{array}{cc}
    1-Q_{1}^{2} & Q_{1}^{2}
    \end{array}\right]\label{eq:CNV_ChoiceQU}\\
    Q_{X|U=0} & =\frac{1}{1-Q_{1}^{2}}\left[\begin{array}{cccc}
    Q_{0}^{2} & Q_{0}Q_{1} & Q_{0}Q_{1} & 0
    \end{array}\right] \label{eq:CNV_ChoiceQX1} \\
    Q_{X|U=1} & =\left[\begin{array}{cccc} 
    0 & 0 & 0 & 1
    \end{array}\right].\label{eq:CNV_ChoiceQX2}
\end{align}
This choice yields an $X$-marginal $Q_{X}$ precisely given by \eqref{eq:CNV_OptQ2},
and it is motivated by the empirical observation from \cite{JournalMU}
that choices of $Q_{UX}$ where $Q_{X|U=1}$ and $Q_{X|U=2}$ have
disjoint supports tend to provide good rates.  We let
the single-letter distribution $Q=[Q_{0}\,\, Q_{1}]$ be 
\begin{equation}
    Q=\left[\begin{array}{cc}
    0.749 & 0.251
    \end{array}\right].\label{eq:CNV_OptQ_SC} 
\end{equation}
which we chose based on a simple brute force search (see Figure \ref{fig:CNV_MainPlot}).  
Note that this choice is similar to that in \eqref{eq:CNV_OptQ}, 
but not identical.

One may question whether the choice of the supports of $Q_{X|U=0}$
and $Q_{X|U=1}$ in \eqref{eq:CNV_ChoiceQX1}--\eqref{eq:CNV_ChoiceQX2}
is optimal.  For example, a similar construction might
set $Q_U(0) = Q_0^2 + Q_0Q_1$, and then replace \eqref{eq:CNV_ChoiceQX1}--\eqref{eq:CNV_ChoiceQX2}
by normalized versions of $[Q_0^2~Q_0Q_1~0~0]$ and $[0~0~Q_0Q_1~Q_1^2]$.
However, after performing a brute force search over the possible support
patterns (there are no more than $2^4$, and many can be ruled out
by symmetry considerations), we found the above pattern to be the
only one to give an improvement on $\LM$, at least for the choices of input
distribution in \eqref{eq:CNV_OptQ} and \eqref{eq:CNV_OptQ_SC}.  
In fact, even after setting
$|\Uc|=3$, considering the third-order product channel $(W^{(3)},q^{(3)})$,
and performing a similar brute force search over the support patterns 
(of which there are no more than $3^8$), we were unable to obtain
an improvement on \eqref{eq:CNV_BoundSC2}.

\subsubsection{Choices of Optimization Parameters}

We now specify the choices of the dual parameters in \eqref{eq:SC_R1_Dual}--\eqref{eq:SC_Rsum_Dual}.
In Appendix \ref{sec:CNV_FURTHER_TECHNIQUES}, we give details of how these parameters
were obtained. We claim that the choice
\begin{equation}
    (R_{0},R_{1})=(0.0356005,0.2403966)\label{eq:CNV_R1R2}
\end{equation}
is permitted; observe that summing these two values and dividing by
two (since we are considering the product channel) yields \eqref{eq:CNV_BoundSC2}.
These values can be verified by setting the parameters as follows:
On the right-hand side of \eqref{eq:SC_R1_Dual}, set 
\begin{align}
    s & =9.4261226 \label{eq:CNV_s1} \\
    \av & =\left[\begin{array}{cccc}
    0.4817048 & -0.2408524 & -0.2408524 & 0\\
    0 & 0 & 0 & 0
    \end{array}\right], \label{eq:CNV_a1}
\end{align}
and on the right-hand side of \eqref{eq:SC_Rsum_Dual}, set
\begin{align}
    \rho_{1} &= 0.7587516 \\
    s &= 9.3419338 \label{eq:CNV_s2} \\
    \av & =\left[\begin{array}{cccc}
    0.7186926 & -0.0488036 & -0.0488036 & 0\\
    0 & 0 & 0 & -0.6210855
    \end{array}\right].\label{eq:CNV_a2}
\end{align}
Once again, we evaluated \eqref{eq:SC_R1_Dual}--\eqref{eq:SC_Rsum_Dual} using
Mathematica's arbitrary-precision arithmetic framework \cite{MathematicaPrecision},
thus ensuring the validity of \eqref{eq:CNV_BoundSC2}.
See \cite{PaperBI_Code} for the relevant C and Mathematica code.
 
\section{Conclusion and Discussion}

We have used our numerical findings, along with an analysis of the
gap to suboptimality for slightly suboptimal input distributions,
to show that it is possible for $\CM$ to exceed $\CLM$
even for binary-input mismatched DMCs. This is in contrast with the
claim in \cite{ConverseMM} that $\CM=\CLM$
for such channels.

An interesting direction for future research is to find a purely theoretical
proof of Counter-Example \ref{prop:CNV_MainResult}; the
non-concavity of $\LM(Q)$ observed in Figure \ref{fig:CNV_MainPlot} may play
a role in such an investigation. Furthermore, it would be of significant
interest to develop a better understanding of \cite{ConverseMM},
including which parts may be incorrect, under what conditions the
converse remains valid, and in the remaining cases, whether a valid
converse lying in between the LM rate and matched capacity can be inferred.

\appendices

\section{Derivations of the Superposition Coding Rates} \label{sec:CNV_DERIVATIONS}

Here we outline how the superposition coding rates in \eqref{eq:CNV_R1SC}--\eqref{eq:SC_Rsum_Dual}  
are obtained using the techniques of \cite{MMSomekh,JournalMU}.  The equivalence of the primal and dual formulations
can also be proved using techniques presented in \cite{JournalMU}.

\subsection{Preliminary Definitions and Results}

The parameters of the random-coding ensemble are a finite auxiliary alphabet $\Uc$, an auxiliary 
codeword distribution $P_{\Uv}$, and a conditional codeword distribution $P_{\Xv|\Uv}$. 
An auxiliary codebook $\{\Uv^{(i)}\}_{i=1}^{M_{0}}$ 
with $M_{0}$ codewords is generated, with each auxiliary codeword
independently distributed according to $P_{\Uv}$.  For each $i=1,\dotsc,M_{0}$, a codebook
$\{\Xv^{(i,j)}\}_{j=1}^{M_{1}}$ with $M_{1}$ codewords is
generated, with each codeword conditionally independently distributed according 
to $P_{\Xv|\Uv}(\cdot|\Uv^{(i)})$. 
The message $m$ at the input to the encoder is indexed as
$(m_{0},m_{1})$, and for any such pair, the corresponding transmitted codeword
is $\Xv^{(m_{0},m_{1})}$.  Thus, the overall number of messages is $M=M_{1}M_{2}$,
yielding a rate of $R=R_{1}+R_{2}$.  More compactly, we have
\begin{equation}
    \bigg\{\Big(\Uv^{(i)},\{\Xv^{(i,j)}\}_{j=1}^{M_{1}}\Big)\bigg\}_{i=1}^{M_{0}} \sim \prod_{i=1}^{M_{0}}\bigg(P_{\Uv}(\uv^{(i)})\prod_{j=1}^{M_{1}}P_{\Xv|\Uv}(\xv^{(i,j)}|\uv^{(i)})\bigg).
\end{equation}
We assume without loss of generality that message $(1,1)$
is transmitted, and we write $\Uv$ and $\Xv$ in place of
$\Uv^{(1)}$ and $\Xv^{(1,1)}$ respectively. We write 
$\Xvtilde$ to denote an arbitrary codeword $\Xv^{(1,j)}$ with $j\ne1$.
Moreover, we let $\Uvbar$ denote an arbitrary auxiliary codeword
$\Uv^{(i)}$ with $i \ne 1$, let $\Xvbar^{(j)}$ ($j=1,\cdots,M_1$) denote the corresponding
codeword $\Xv^{(i,j)}$, and let $\Xvbar$ denote $\Xvbar^{(j)}$ for
an arbitrary choice of $j$. Thus, defining $\Yv$ to be the channel
output, we have
\begin{equation}
    (\Uv,\Xv,\Yv,\Xvtilde,\Uvbar,\Xvbar) \sim P_{\Uv}(\uv)P_{\Xv|\Uv}(\xv|\uv)W^{n}(\yv|\xv)P_{\Xv|\Uv}(\xvtilde|\uv)P_{\Uv}(\uvbar)P_{\Xv|\Uv}(\xvbar|\uvbar). \label{eq:SC_VecDistr}
\end{equation}
The decoder estimates $\hat{m}=(\hat{m}_{0},\hat{m}_{1})$
according to the decoding rule in \eqref{eq:CNV_DecodingRule}. We define the following error events:
\begin{samepage}
\begin{tabbing}
    ~~~{\emph{(Type 0)}}~~~ \= $q^n(\Xv^{(i,j)},\Yv) \ge q^n(\Xv,\Yv)$ for some $i \ne 1$, $j$; \\
    ~~~{\emph{(Type 1)}}~~~ \> $q^n(\Xv^{(1,j)},\Yv) \ge q^n(\Xv,\Yv)$ for some $j \ne 1$.
\end{tabbing}
\end{samepage}
The probabilities of these events are
denoted by $\peobar(n,M_{0},M_{1})$ and $\peibar(n,M_{1})$ respectively.
The overall random-coding error probability $\pebar(n,M_{0},M_{1})$ 
clearly satisfies $\pebar \le \peobar + \peibar$.

We begin by deriving the following non-asymptotic bounds:
\begin{gather}
    \peobar(n,M_{0},M_{1}) \le \EE\left[\min\left\{ 1,(M_{0}-1)\EE\Bigg[\min\bigg\{1,M_{1}\PP\bigg[\frac{q^{n}(\Xvbar,\Yv)}{q^{n}(\Xv,\Yv)}\ge1 \,\Big|\, \Uvbar\bigg]\bigg\}\,\bigg|\,\Uv,\Xv,\Yv\Bigg]\right\} \right] \label{eq:SC_Bound0} \\
    \peibar(n,M_{1}) \le \EE\Bigg[\min\bigg\{1,(M_1-1)\PP\bigg[\frac{q^{n}(\Xvtilde,\Yv)}{q^{n}(\Xv,\Yv)} \ge 1 \,\Big|\, \Uv,\Xv,\Yv\bigg]\bigg\}\Bigg]. \label{eq:SC_Bound1}         
\end{gather}
We will see that \eqref{eq:SC_Bound0} recovers the rate conditions in \eqref{eq:CNV_RsumSC} and
\eqref{eq:SC_Rsum_Dual}, whereas \eqref{eq:SC_Bound1} recovers those in \eqref{eq:CNV_R1SC} and
\eqref{eq:SC_R1_Dual}.  We focus on the type-0 event throughout the appendix,
since the type-1 event is simpler, and is handled using standard techniques associated with
the case of independent codewords.

To derive \eqref{eq:SC_Bound0}, we first note that
\begin{equation}
    \peobar = \PP\bigg[\bigcup_{i\ne1,j\ne1}\Big\{\frac{q^{n}(\Xv^{(i,j)},\Yv)}{q^{n}(\Xv,\Yv)} \ge 1\Big\}\bigg].
\end{equation}
Writing the probability as an expectation given $(\Uv,\Xv,\Yv)$ and applying the truncated union
bound to the union over $i$, we obtain
\begin{equation}
    \peobar \le \EE\Bigg[\min\bigg\{1,(M_0-1)\PP\bigg[\bigcup_{j\ne1}\Big\{\frac{q^{n}(\Xvbar^{(j)},\Yv)}{q^{n}(\Xv,\Yv)} \ge 1\Big\} \,\bigg|\, \Uv,\Xv,\Yv\bigg]\bigg\}\Bigg].
\end{equation}
Applying the same argument to the union over $j$, we obtain the desired upper bound.

Before proceeding, we introduce some standard notation and terminology associated
with the method of types (e.g. see \cite[Ch. 2]{CsiszarBook}).  For a given joint  
type, say $\Ptilde_{UX}$, the type class $T^{n}(\Ptilde_{UX})$ is 
defined to be the set of all sequences in $\Uc^n \times \Xc^{n}$ with type $\Ptilde_{UX}$.  Similarly,
for a given joint type $\Ptilde_{UX}$ and sequence $\uv \in T^{n}(\Ptilde_{U})$, the conditional type
class $T_{\uv}^{n}(\Ptilde_{UX})$ is defined to be the set of all sequences
$\xv$ such that $(\uv,\xv) \in T^{n}(\Ptilde_{UX})$.

In the remainder of the appendix, we consider the constant-composition ensemble, described by
\begin{align}
    P_{\Uv}(\uv)         &= \frac{1}{|T^{n}(Q_U)|} \openone\Big\{\uv\in T^{n}(Q_U)\Big\} \label{eq:SC_DistrU_CC} \\ 
    P_{\Xv|\Uv}(\xv|\uv) &= \frac{1}{|T_{\uv}^{n}(Q_{X|U})|} \openone\Big\{\xv\in T_{\uv}^{n}(Q_{X|U})\Big\}. \label{eq:SC_DistrX_CC}
\end{align}
Here we have assumed that $Q_{UX}$ is a joint type for notational convenience; more generally, we 
can approximate $Q_{UX}$ by a joint type and the analysis is unchanged.

\subsection{Derivation of the Primal Expression}

We will derive \eqref{eq:CNV_RsumSC} by showing that the 
error exponent of $\peobar$ (i.e. $\liminf_{n\to\infty}-\frac{1}{n}\log\peobar$) is lower bounded by
\begin{multline}
    \Eorcc(Q_{UX},R_{0},R_{1})\defeq\min_{P_{UXY} \,:\, P_{UX}=Q_{UX}} \min_{\substack{\Ptilde_{UXY} \,:\,\Ptilde_{UX}=Q_{UX},\,\Ptilde_{Y}=P_{Y} \\ \EE_{\Ptilde}[\log q(X,Y)]\ge\EE_{P}[\log q(X,Y)]}} \\
        D(P_{UXY}\|Q_{UX}\times W)+\Big[I_{\Ptilde}(U;Y)+\big[I_{\Ptilde}(X;Y|U)-R_{1}\big]^{+}-R_{0}\Big]^{+}. \label{eq:SC_Er0_CC}
\end{multline}
The objective is always positive when $P_{UXY}$ is bounded away from $Q_{UX} \times W$.
Hence, and by applying a standard continuity argument as in \cite[Lemma1]{Csiszar1},
we may substitute $P_{UXY} = Q_{UX} \times W$ to obtain the desired rate condition in \eqref{eq:CNV_RsumSC}.

We obtain \eqref{eq:SC_Er0_CC} by analyzing \eqref{eq:SC_Bound0} using the method of types. 
Except where stated otherwise, it should be understood that unions, summations, and minimizations
over joint distributions (e.g. $P_{UXY}$) are only over joint types, rather than being over
all distributions.

Let us first condition on $\Uv=\uv$, $\Xv=\xv$, $\Yv=\yv$ and $\Uvbar=\uvbar$ being fixed sequences, and
let $P_{UXY}$ and $\Phat_{UY}$ respectively denote the joint types of $(\uv,\xv,\yv)$
and $(\uvbar,\yv)$.  We can write the inner probability in \eqref{eq:SC_Bound0} as
\begin{align}
    \PP\bigg[\bigcup_{\substack{\Ptilde_{UXY} \,:\,\Ptilde_{UX}=Q_{UX},\,\Ptilde_{UY}=\Phat_{UY} \\ \EE_{\Ptilde}[\log q(X,Y)]\ge\EE_{P}[\log q(X,Y)]}}\Big\{ (\uvbar,\Xvbar,\yv) \in T^n(\Ptilde_{UXY}) \Big\} \,\Big|\, \Uvbar=\uvbar \bigg]. \label{eq:CNV_PrimalTerm1}
\end{align}
Note that the constraint $\Ptilde_{UX}=Q_{UX}$ arises
since every $(\uv,\xv)$ pair has joint type $Q_{UX}$ by construction.
Applying the union bound, the property of types 
$\PP[(\uvbar,\Xvbar,\yv) \in T^n(\Ptilde_{UXY}) \,|\, \Uvbar=\uvbar] \le e^{-n I_{\Ptilde}(X;Y|U)}$ \cite[Ch.~2]{CsiszarBook},
and the fact that the number of joint types is polynomial in $n$, we see that 
the negative normalized (by $\frac{1}{n}$) logarithm of 
\eqref{eq:CNV_PrimalTerm1} is lower bounded by
\begin{equation}
    \min_{\substack{\Ptilde_{UXY} \,:\,\Ptilde_{UX}=Q_{UX},\,\Ptilde_{UY}=\Phat_{UY} \\ \EE_{\Ptilde}[\log q(X,Y)]\ge\EE_{P}[\log q(X,Y)]}} I_{\Ptilde}(X;Y|U) \label{eq:CNV_PrimalExp1}
\end{equation} 
plus an asymptotically vanishing term.

Next, we write the inner expectation in \eqref{eq:SC_Bound0} (conditioned on $\Uv=\uv$, $\Xv=\xv$ and $\Yv=\yv$) as
\begin{multline}
    \sum_{\Phat_{UY} \,:\, \Phat_{U} = Q_{U},\Phat_{Y}=P_{Y}} \PP\Big[(\Uvbar,\yv) \in T^n(\Phat_{UY}) \Big]\min\bigg\{1, \\
     M_1\PP\bigg[\bigcup_{\substack{\Ptilde_{UXY} \,:\,\Ptilde_{UX}=Q_{UX},\,\Ptilde_{UY}=\Phat_{UY} \\ \EE_{\Ptilde}[\log q(X,Y)]\ge\EE_{P}[\log q(X,Y)]}}\Big\{ (\uvbar,\Xvbar,\yv) \in T^n(\Ptilde_{UXY}) \Big\} \,\Big|\, \Uvbar=\uvbar \bigg]\bigg\}, \label{eq:CNV_PrimalTerm2}
\end{multline}
where now $\uvbar$ denotes an arbitrary sequence such that $(\uv,\yv) \in T^n(\Phat_{UY})$.
Combining \eqref{eq:CNV_PrimalExp1} with the property of types 
$\PP[(\Uvbar,\yv) \in T^n(\Phat_{UY}) ] \le e^{-n I_{\Phat}(U;Y)}$ \cite[Ch.~2]{CsiszarBook}, we see that
the negative normalized logarithm of \eqref{eq:CNV_PrimalTerm2} is lower bounded by
\begin{equation}
    \min_{\Phat_{UY} \,:\, \Phat_{U} = Q_{U},\Phat_{Y}=P_{Y}} I_{\Phat}(U;Y) + \min_{\substack{\Ptilde_{UXY} \,:\,\Ptilde_{UX}=Q_{UX},\,\Ptilde_{UY}=\Phat_{UY} \\ \EE_{\Ptilde}[\log q(X,Y)]\ge\EE_{P}[\log q(X,Y)]}} \big[I_{\Ptilde}(X;Y|U) - R_1\big]^+ \label{eq:CNV_PrimalExp2}
\end{equation} 
plus an asymptotically vanishing term.
Combining the two minimizations into one via the constraint $\Ptilde_{UY}=\Phat_{UY}$, 
we see that the right-hand side of \eqref{eq:CNV_PrimalExp2} coincides with
that of \eqref{eq:CNV_RsumSC} (note, however, that $P_{UXY}$ need not equal 
$Q_{UX} \times W$ at this stage).  Finally, the derivation
of \eqref{eq:SC_Er0_CC} is concluded by handling the outer expectation in \eqref{eq:SC_Bound0} in 
the same way as the inner one, applying the property $\PP[(\Uv,\Xv,\Yv) \in T^n(P_{UXY})] \le e^{-n D(P_{UXY}\|Q_{UX}\times W)}$ \cite[Ch.~2]{CsiszarBook}
(which follows since $(\Uv,\Xv)$ is uniform on $T^n(Q_{UX})$),
and expanding the minimization set from joint types to general joint distributions.

\subsection{Derivation of the Dual Expression}

Expanding \eqref{eq:SC_Bound0} and applying Markov's
inequality and $\min\{1,\alpha\}\le\alpha^{\rho}$ ($\rho\in[0,1]$),
we obtain 
\begin{equation}
    \peobar \le \sum_{\uv,\xv,\yv}P_{\Uv}(\uv)P_{\Xv|\Uv}(\xv|\uv) W^{n}(\yv|\xv)
        \Bigg(M_{0}\sum_{\uvbar}P_{\Uv}(\uvbar)\bigg(M_{1}\frac{\sum_{\xvbar}P_{\Xv|\Uv}(\xvbar|\uvbar)q^{n}(\xvbar,\yv)^{s}}{q^{n}(\xv,\yv)^{s}}\bigg)^{\rho_{1}}\Bigg)^{\rho_{0}}. \label{eq:CNV_DualDeriv1}    
\end{equation}
for any $\rho_{0}\in[0,1]$, $\rho_{1}\in[0,1]$ and $s\ge0$.  Let $a(u,x)$ be an
arbitrary function on $\Uc\times\Xc$, and let $a^n(\uv,\xv) \triangleq \sum_{i=1}^{n}a(u_i,x_i)$
be its additive $n$-letter extension.  Since $(\Uv,\Xv)$ and $(\Uvbar,\Xvbar)$ have the
same joint type (namely, $Q_{UX}$) by construction, we have $a^n(\uv,\xv) = a^n(\uvbar,\xvbar)$,
and hence we can write \eqref{eq:CNV_DualDeriv1} as
\begin{equation}
    \peobar \le \sum_{\uv,\xv,\yv}P_{\Uv}(\uv)P_{\Xv|\Uv}(\xv|\uv) W^{n}(\yv|\xv)
        \Bigg(M_{0}\sum_{\uvbar}P_{\Uv}(\uvbar)\bigg(M_{1}\frac{\sum_{\xvbar}P_{\Xv|\Uv}(\xvbar|\uvbar)q^{n}(\xvbar,\yv)^{s}e^{a^n(\uvbar,\xvbar)}}{q^{n}(\xv,\yv)^{s}e^{a^n(\uv,\xv)}}\bigg)^{\rho_{1}}\Bigg)^{\rho_{0}}. \label{eq:CNV_DualDeriv1}    
\end{equation}
Upper bounding each constant-composition distribution by a polynomial factor times 
the corresponding i.i.d. distribution (i.e.~$P_{\Uv}(\uv) \le (n+1)^{|\Uc|-1}\prod_{i=1}^{n}Q_U(u_i)$ \cite[Ch.~2]{CsiszarBook}),
we see the exponent of $\peobar$ is lower bounded by
\begin{equation}
    \max_{\rho_0\in[0,1],\rho_1\in[0,1]} E_0(Q_{UX},\rho_0,\rho_1) - \rho_0(R_0 + \rho_1 R_1),
\end{equation}
where
\begin{equation}
    E_0(Q_{UX},\rho_0,\rho_1) \triangleq \sup_{s\ge0,a(\cdot,\cdot)} -\log\sum_{u,x}Q_{UX}(u,x)W(y|x)\Bigg( \sum_{\ubar}Q_{U}(\ubar)\bigg(\frac{\sum_{\xbar}Q_{X|U}(\xbar|\ubar)q(\xbar,y)^s e^{a(\ubar,\xbar)}}{q(x,y)^s e^{a(u,x)}}\bigg)^{\rho_1} \Bigg)^{\rho_0}. \label{eq:CNV_E0}
\end{equation}
We obtain \eqref{eq:SC_Rsum_Dual} in the same way as Gallager's single-user analysis \cite[Sec. 5.6]{Gallager}
by evaluating the partial derivative of the objective in \eqref{eq:CNV_E0} at $\rho_0=0$ (see also 
\cite{Variations} for the corresponding approach to deriving the LM rate).

\section{Further Numerical Techniques Used \label{sec:CNV_FURTHER_TECHNIQUES}}

In this section, we present further details
of our numerical techniques for the sake of reproducibility.
Except where stated otherwise, the implementations were done
in C.  Our code is available online \cite{PaperBI_Code}.

The algorithms here do not play a direct role in 
establishing Counter-Example \ref{prop:CNV_MainResult}.  We thus
resort to ``ad-hoc'' approaches with manually-tuned
parameters.  In particular, we make no claims regarding the
convergence of these algorithms or their effectiveness in 
handling channels and decoding metrics differing from
Counter-Example \ref{prop:CNV_MainResult}.

\subsection{Evaluating $\LM(Q)$ via the Dual Expression and
Gradient Descent}

Here we describe how we optimized the parameters in \eqref{eq:CNV_DualLM}
for a fixed input distribution $Q$ to produce the dual values plotted in
Figure \ref{fig:CNV_MainPlot}. Note that replacing the optimization
by fixed values leads to a lower bound, whereas for the primal expression
it led to an upper bound. Thus, since the two are very close in Figure
\ref{fig:CNV_MainPlot}, we can be assured that the true value of
$\LM(Q)$ has been characterized accurately, at least
for the values of $Q$ shown. While we focus on the binary-input setting
here, the techniques can be applied to an arbitrary mismatched DMC.
For brevity, we write $a_{x}\defeq a(x)$.

Let $Q$ be given, and let $I(\vv)$ be the corresponding 
objective in \eqref{eq:CNV_DualLM} as a function of $\vv\defeq[s\,\, a_{0}\,\, a_{1}]^{T}$.
Moreover, let $\nabla I(\vv)$
denote the $3\times1$ corresponding gradient vector containing the
partial derivatives of $I(\cdot)$. These are all easily evaluated in
closed form by a direct differentiation. We used the following standard
gradient descent algorithm, which depends on the initial values $(s^{(0)},a_{0}^{(0)},a_{1}^{(0)})$,
a sequence of step sizes $\{t^{(i)}\}$, and a termination parameter $\epsilon$:
\begin{enumerate}
\item Set $i=0$ and initialize $\vv^{(0)}=[s^{(0)}\,\, a_{0}^{(0)}\,\, a_{1}^{(0)}]^{T}$;
\item Set $\vv^{(i+1)}=\vv^{(i)}-t^{(i)}\nabla I(\vv^{(i)})$;
\item If $\|\nabla I(\vv^{(i+1)})\|\le\epsilon$ then terminate;
otherwise, increment $i$ and return to Step 2.
\end{enumerate}
We used the initial parameters $(s^{(0)},a_{0}^{(0)},a_{1}^{(0)})=(1,0,0)$,
a fixed step size $t^{(i)}=1$, and a termination parameter $\epsilon=10^{-6}$.

Note that we have ignored the constraint $s \ge 0$, but this has no effect
on the maximization.  This is seen by noting that $s$ is a Lagrange multiplier
corresponding to the constraint on the metric in \eqref{eq:CNV_PrimalLM},
and the inequality therein can be replaced by an equality as long as
$\LM(Q) > 0$ \cite[Lemma 1]{Csiszar2}.  The equality constraint yields
a Lagrange multiplier on $\RR$, rather than $\RR_{+}$.

\subsection{Evaluating $\ISC^{(2)}(Q_{UX})$ via the Dual Expression
and Gradient Descent}

To obtain the dual curve for $\ISC^{(2)}$ in Figure \ref{fig:CNV_MainPlot},
we optimized the parameters in \eqref{eq:SC_R1_Dual}--\eqref{eq:SC_Rsum_Dual}
in a similar fashion to the previous subsection.  In fact,
the optimization of \eqref{eq:SC_R1_Dual} was done exactly as above,
with the same initial parameters (i.e. initializing $s=1$ and $a(u,x)=0$
for all $(u,x)$).  By letting \eqref{eq:SC_R1_Dual} hold with equality,
the solution to this optimization gives a value for $R_1$.

Handling the optimization in \eqref{eq:SC_Rsum_Dual} was less 
straightforward. We were unable to verify the \emph{joint} 
concavity of the objective in $(\rho_{1},s,a)$, and we in fact found a naive
application of gradient descent to be problematic due to overly
large changes in $\rho_1$ on each step.  Moreover, while it is safe to ignore
the constraint $s\ge0$ in the same way as the previous subsection,
the same is not true of the constraint $\rho_1 \in [0,1]$.
We proceed by presenting a modified algorithm that handles these issues.

Similarly to the previous subsection, we let $\vv$ be the vector
of parameters, let $I_0(\vv)$ denote the objective in \eqref{eq:SC_Rsum_Dual}
with $Q_{UX}$ and $R_1$ fixed (the latter chosen as the value
given by the evaluation of  \eqref{eq:SC_R1_Dual}), and let $\nabla I_0(\vv)$ be the 
corresponding gradient vector.  Moreover, we define
\begin{equation}
    \Phi(\rho_1) \triangleq \begin{cases} 0 & \rho_1 < 0 \\ \rho_1 & \rho_1 \in [0,1] \\ 1 & \rho_1 > 1. \end{cases}
\end{equation}
Finally, we let $\vv_{{-\rho_1}}$ denote the vector
$\vv$ with the entry corresponding to $\rho_1$ removed,
and similarly for other vectors (e.g. $(\nabla I_0(\vv))_{{-\rho_1}}$).

We applied the following variation of gradient descent, which depends
on the initial parameters, the step sizes $\{t^{(i)}\}$, and two
parameters $\epsilon$ and $\epsilon'$:
\begin{enumerate}
\item Set $i=0$ and initialize $\vv^{(0)}$.
\item Set $\vv_{-\rho_1}^{(i+1)} = \vv_{-\rho_1}^{(i)}-t^{(i)}(\nabla I(\vv^{(i)}))_{-\rho_1}$.
\item If $\|(\nabla I_0(\vv^{(i)}))_{{-\rho_1}}\| \le \epsilon'$ then 
set $\rho_1^{(i+1)} = \Phi\big(\rho_1^{(i)} - t^{(i)} \frac{\partial I_0}{\partial \rho_1}\big|_{\vv = \vv^{(i)}}\big)$; 
otherwise set $\rho_1^{(i+1)} = \rho_1^{(i)}$.
\item Terminate if either of the following conditions hold: 
(i) $\|\nabla I(\vv^{(i+1)})\|\le\epsilon$; (ii) $\|(\nabla I_0(\vv^{(i+1)}))_{{-\rho_1}}\| \le \epsilon$
and $\rho_1^{(i+1)} \in \{0,1\}$.  Otherwise, increment $i$ and return to Step 2.
\end{enumerate}
In words, $\rho_1$ is only updated if the norm of the gradient corresponding
to $(s,a)$ is sufficiently small, and the algorithm may terminate when $\rho_1$
saturates to one of the two endpoints of $[0,1]$ (rather than arriving at a
local maximum).  
We initialized $s$ and $\rho_1$ to $1$, and each $a(u,x)$ to zero.  We again
used a constant step size $t^{(i)}=1$, and we chose the parameters
$\epsilon = 10^{-6}$ and $\epsilon' = 10^{-2}$. 

\subsection{Evaluating $\ISC^{(2)}(Q_{UX})$ via the Primal Expression}

Since we only computed the primal expression for $\ISC^{(2)}(Q_{UX})$ with a relatively small
number of input distributions (namely, those shown in Figure \ref{fig:CNV_MainPlot}),
computational complexity was a minor issue, so we resorted to the general-purpose software
CVX for MATLAB \cite{CVX}. In the same way as the previous subsection, we solved
the right-hand side of \eqref{eq:CNV_R1SC} to find $R_{1}$, then
substituted the resulting value into \eqref{eq:CNV_RsumSC} to find $R_{0}$.

\section{Achievability of \eqref{eq:CNV_BoundSC2} via Expurgated Parallel Coding} \label{sec:CNV_MAC_RATE}

Here we outline how the achievable rate of $0.137998$ nats/use in \eqref{eq:CNV_BoundSC2} can be 
obtained using Lapidoth's expurgated parallel coding rate.  We verified this
value by evaluating the primal expressions in \cite{MacMM} using CVX \cite{CVX},
and also by evaluating the equivalent dual expressions in \cite{JournalMU}
by a suitable adaptation of the dual optimization parameters for superposition
coding given in Section \ref{sub:CNV_SC_EVAL}.  We focus our attention on
the latter, since it immediately provides a concrete lower bound even when
the optimization parameters are slightly suboptimal.

The parameters to Lapidoth's rate are two finite alphabets $\Xc_1$ and $\Xc_2$, 
two corresponding input distributions $Q_1$ and $Q_2$, and a function 
$\phi(x_1,x_2)$ mapping $\Xc_1$ and $\Xc_2$ to the channel input alphabet.
For any such parameters, the rate $R = R_1+R_2$ is achievable provided that \cite{JournalMU,Thesis}
\begin{align}
    R_{1} &\le \sup_{s\ge0,a(\cdot,\cdot)}\EE\left[\log\frac{q(\phi(X_{1},X_{2}),Y)^{s}e^{a(X_{1},X_{2})}}{\EE\big[q(\phi(\Xbar_{1},X_{2}),Y)^{s}e^{a(\Xbar_{1},X_{2})}\,|\, X_{2},Y\big]}\right]\label{eq:MAC_ExR1_Gen} \\
    R_{2} &\le \sup_{s\ge0,a(\cdot,\cdot)}\EE\left[\log\frac{q(\phi(X_{1},X_{2}),Y)^{s}e^{a(X_{1},X_{2})}}{\EE\big[q(\phi(X_{1},\Xbar_{2}),Y)^{s}e^{a(X_{1},\Xbar_{2})}\,|\, X_{1},Y\big]}\right],\label{eq:MAC_ExR2_Gen}
\end{align}
and at least one of the following holds:
\begin{align}
    R_{1} &\le \sup_{\rho_{2}\in[0,1],s\ge0,a(\cdot,\cdot)} \EE\left[\log\frac{\big(q(\phi(X_{1},X_{2}),Y)^{s}e^{a(X_{1},X_{2})}\big)^{\rho_{2}}}{\EE\Big[\Big(\EE\big[q(\phi(\Xbar_{1},\Xbar_{2}),Y)^{s}e^{a(\Xbar_{1},\Xbar_{2})}\,\big|\,\Xbar_{1}\big]\Big)^{\rho_{2}}\,\big|\,Y\Big]}\right] - \rho_{2}R_{2} \label{eq:MAC_ExR12_1_Gen} \\
    R_{2} &\le \sup_{\rho_{1}\in[0,1],s\ge0,a(\cdot,\cdot)} \EE\left[\log\frac{\big(q(\phi(X_{1},X_{2}),Y)^{s}e^{a(X_{1},X_{2})}\big)^{\rho_{1}}}{\EE\Big[\Big(\EE\big[q(\phi(\Xbar_{1},\Xbar_{2}),Y)^{s}e^{a(\Xbar_{1},\Xbar_{2})}\,\big|\,\Xbar_{2}\big]\Big)^{\rho_{1}}\,\big|\,Y\Big]}\right] - \rho_{1}R_{1}, \label{eq:MAC_ExR12_2_Gen}
\end{align}
where $(X_{1},X_{2},Y,\Xbar_{1},\Xbar_{2})\sim Q_{1}(x_{1})Q_{2}(x_{2})W(y|\phi(x_{1},x_{2}))Q_{1}(\xbar_{1})Q_{2}(\xbar_{2})$.

Recall the input distribution $Q_{UX}$ for superposition coding on the
second-order product channel given in \eqref{eq:CNV_ChoiceQU}--\eqref{eq:CNV_ChoiceQX2}.
Denoting the four inputs of the product channel as $\{(0,0),(0,1),(1,0),(1,1)\}$, 
we set $\Xc_1 = \{(0,0),(0,1),(1,0)\}$, $\Xc_2 = \Uc = \{0,1\}$, and
\begin{align}
    Q_{X_1} & =\frac{1}{1-Q_{1}^{2}}\left[\begin{array}{cccc}
    Q_{0}^{2} & Q_{0}Q_{1} & Q_{0}Q_{1}
    \end{array}\right] \\
    Q_{X_2} & =\left[\begin{array}{cc}
    1-Q_{1}^{2} & Q_{1}^{2}
    \end{array}\right] \\
    \phi(x_1,x_2) &= \begin{cases} x_1 & x_2 = 0 \\ (1,1) & x_2 = 1. \end{cases}
\end{align}
This induces a joint distribution $Q_{X_1X_2X}(x_1,x_2,x) = Q_{X_1}(x_1)Q_{X_2}(x_2)\openone\{x = \phi(x_1,x_2)\}$.
The idea behind this choice is that the marginal distribution
$Q_{X_2X}$ coincides with our choice of $Q_{UX}$ for SC.  

By the structure of our input distributions, there is in fact a one-to-one
correspondence between $(u,x)$ and $(x_1,x_2)$, thus allowing us
to immediately use the dual parameters $(s,a,\rho_1)$ from SC
for the expurgated parallel coding rate.  More precisely, using 
the superscripts $(\cdot)^{\mathrm{sc}}$ and $(\cdot)^{\mathrm{ex}}$ 
to distinguish between the two ensembles, we set
\begin{align}
    R_1^{\mathrm{ex}} &= R_1^{\mathrm{sc}} \\
    R_2^{\mathrm{ex}} &= R_0^{\mathrm{sc}} \\
    s^{\mathrm{ex}} &= s^{\mathrm{sc}} \\
    a^{\mathrm{ex}}(x_1,x_2) &= a^{\mathrm{sc}}(x_2,\phi(x_1,x_2)) \\
    \rho_1^{\mathrm{ex}} &= \rho_1^{\mathrm{sc}}.
\end{align}  
Using these identifications along with the choices of the superposition
coding parameters in \eqref{eq:CNV_s1}--\eqref{eq:CNV_a2}, we verified
numerically that the right-hand side of \eqref{eq:MAC_ExR1_Gen} 
(respectively, \eqref{eq:MAC_ExR12_2_Gen}) coincides with that of
\eqref{eq:SC_R1_Dual} (respectively, \eqref{eq:SC_Rsum_Dual}).
Finally, to conclude that the expurgated parallel coding rate recovers
\eqref{eq:CNV_BoundSC2}, we numerically verified that the rate $R_2$ 
resulting from \eqref{eq:MAC_ExR1_Gen} and \eqref{eq:MAC_ExR12_2_Gen}
(which, from \eqref{eq:CNV_R1R2}, is $0.0356005$) also satisfies
\eqref{eq:MAC_ExR2_Gen}.  In fact, the inequality is strict, with the 
right-hand side of \eqref{eq:MAC_ExR2_Gen} being at least $0.088$. 

\bibliographystyle{IEEEtran}
\bibliography{12-Paper,18-MultiUser,18-SingleUser,35-Other}

\end{document}